\documentclass[11pt,a4paper,reqno] %,intlimits,sumlimits]
{amsart}

\usepackage{todonotes}
\usepackage[round]{natbib}
\usepackage{amsmath}
\usepackage{amssymb}
\usepackage[mathscr]{euscript}
\usepackage{enumerate}
\usepackage{xspace}
\usepackage{color}
\usepackage{enumerate}
\usepackage{enumitem}
\usepackage{amsthm}
\usepackage[draft]{pdfcomment}
\DeclareMathAlphabet{\mathpzc}{OT1}{pzc}{m}{it}

\begin{document}

\theoremstyle{plain}
\newtheorem{theorem}{Theorem}[section]
\newtheorem{lemma}[theorem]{Lemma}
\newtheorem{proposition}[theorem]{Proposition}
\newtheorem{corollary}[theorem]{Corollary}
\newtheorem{claim}[theorem]{Claim}
\newtheorem{definition}[theorem]{Definition}

\theoremstyle{definition}
\newtheorem{remark}[theorem]{Remark}
\newtheorem{note}[theorem]{Note}
\newtheorem{example}[theorem]{Example}
\newtheorem{assumption}[theorem]{Assumption}
\newtheorem*{notation}{Notation}
\newtheorem*{assuL}{Assumption ($\mathbb{L}$)}
\newtheorem*{assuAC}{Assumption ($\mathbb{AC}$)}
\newtheorem*{assuEM}{Assumption ($\mathbb{EM}$)}
\newtheorem*{assuES}{Assumption ($\mathbb{ES}$)}
\newtheorem*{assuM}{Assumption ($\mathbb{M}$)}
\newtheorem*{assuMM}{Assumption ($\mathbb{M}'$)}
\newtheorem*{assuL1}{Assumption ($\mathbb{L}1$)}
\newtheorem*{assuL2}{Assumption ($\mathbb{L}2$)}
\newtheorem*{assuL3}{Assumption ($\mathbb{L}3$)}
\newtheorem{charact}[theorem]{Characterization}

\renewenvironment{proof}{{\parindent 0pt \it{ Proof:}}}{\mbox{}\hfill\mbox{$\Box\hspace{-0.5mm}$}\vskip 16pt}
\newenvironment{proofthm}[1]{{\parindent 0pt \it Proof of Theorem #1:}}{\mbox{}\hfill\mbox{$\Box\hspace{-0.5mm}$}\vskip 16pt}
\newenvironment{prooflemma}[1]{{\parindent 0pt \it Proof of Lemma #1:}}{\mbox{}\hfill\mbox{$\Box\hspace{-0.5mm}$}\vskip 16pt}
\newenvironment{proofcor}[1]{{\parindent 0pt \it Proof of Corollary #1:}}{\mbox{}\hfill\mbox{$\Box\hspace{-0.5mm}$}\vskip 16pt}
\newenvironment{proofprop}[1]{{\parindent 0pt \it Proof of Proposition #1:}}{\mbox{}\hfill\mbox{$\Box\hspace{-0.5mm}$}\vskip 16pt}

\newcommand{\y}{\textcolor{yellow}}
\newcommand{\bl}{\textcolor{blue}}

\newcommand{\Law}{\ensuremath{\mathop{\mathrm{Law}}}}
\newcommand{\loc}{{\mathrm{loc}}}
\newcommand{\Log}{\ensuremath{\mathop{\mathcal{L}\mathrm{og}}}}
\newcommand{\Meixner}{\ensuremath{\mathop{\mathrm{Meixner}}}}

\let\MID\mid
\renewcommand{\mid}{|}

\let\SETMINUS\setminus
\renewcommand{\setminus}{\backslash}

\def\stackrelboth#1#2#3{\mathrel{\mathop{#2}\limits^{#1}_{#3}}}

\renewcommand{\theequation}{\thesection.\arabic{equation}}
\numberwithin{equation}{section}

\newcommand\llambda{{\mathchoice
      {\lambda\mkern-4.5mu{\raisebox{.4ex}{\scriptsize$\backslash$}}}
      {\lambda\mkern-4.83mu{\raisebox{.4ex}{\scriptsize$\backslash$}}}
      {\lambda\mkern-4.5mu{\raisebox{.2ex}{\footnotesize$\scriptscriptstyle\backslash$}}}
      {\lambda\mkern-5.0mu{\raisebox{.2ex}{\tiny$\scriptscriptstyle\backslash$}}}}}

\newcommand{\prozess}[1][L]{{\ensuremath{#1=(#1_t)_{0\le t\le T}}}\xspace}
\newcommand{\prazess}[1][L]{{\ensuremath{#1=(#1_t)_{0\le t\le T^*}}}\xspace}

\newcommand{\g}{\textcolor{green}}
\newcommand{\lijepoa}{{\mathcal{A}}}
\newcommand{\lijepob}{{\mathcal{B}}}
\newcommand{\lijepoc}{{\mathcal{C}}}
\newcommand{\lijepod}{{\mathcal{D}}}
\newcommand{\lijepoe}{{\mathcal{E}}}
\newcommand{\lijepof}{{\mathcal{F}}}
\newcommand{\lijepog}{{\mathcal{G}}}
\newcommand{\lijepok}{{\mathcal{K}}}
\newcommand{\lijepoo}{{\mathcal{O}}}
\newcommand{\lijepop}{{\mathcal{P}}}
\newcommand{\lijepoh}{{\mathcal{H}}}
\newcommand{\lijepom}{{\mathcal{M}}}
\newcommand{\lijepou}{{\mathcal{U}}}
\newcommand{\lijepov}{{\mathcal{V}}}
\newcommand{\lijepoy}{{\mathcal{Y}}}
\newcommand{\cF}{{\mathcal{F}}}
\newcommand{\cG}{{\mathcal{G}}}
\newcommand{\cH}{{\mathcal{H}}}
\newcommand{\cM}{{\mathcal{M}}}
\newcommand{\cD}{{\mathcal{D}}}
\newcommand{\bD}{{\mathbb{D}}}
\newcommand{\bF}{{\mathbb{F}}}
\newcommand{\bG}{{\mathbb{G}}}
\newcommand{\bH}{{\mathbb{H}}}

\newcommand{\er}{{\mathbb{R}}}
\newcommand{\ce}{{\mathbb{C}}}
\newcommand{\erd}{{\mathbb{R}^{d}}}
\newcommand{\en}{{\mathbb{N}}}
\newcommand{\de}{{\mathrm{d}}}
\newcommand{\im}{{\mathrm{i}}}
\newcommand{\set}[1]{\ensuremath{\left\{#1\right\}}}
\newcommand{\indik}{{\mathbf{1}}}
\newcommand{\D}{{\mathbb{D}}}
\newcommand{\E}{{\mathbb{E}}}
\newcommand{\N}{{\mathbb{N}}}
\newcommand{\Q}{{\mathbb{Q}}}
\renewcommand{\P}{{\mathbb{P}}}
\newcommand{\ud}{\operatorname{d}\!}
\newcommand{\ii}{\operatorname{i}\kern -0.8pt}
\newcommand{\Var}{\operatorname{Var }\,}
\newcommand{\dt}{\operatorname{d}\!t}   %\newcommand{\dt}{\mbox{d$t$}}
\newcommand{\ds}{\operatorname{d}\!s}   %{\mbox{d$s$}}
\newcommand{\dy}{\operatorname{d}\!y }    %{\mbox{d$y$}}
\newcommand{\du}{\operatorname{d}\!u}  %{\mbox{d$u$}
\newcommand{\dv}{\operatorname{d}\!v}   %{\mbox{d$v$}}
\newcommand{\dx}{\operatorname{d}\!x}   %{\mbox{d$x$}}
\newcommand{\dq}{\operatorname{d}\!q}   %{\mbox{d$q$}}

\def\EM{\ensuremath{(\mathbb{EM})}\xspace}
\newcommand{\la}{\langle}
\newcommand{\ra}{\rangle}

\newcommand{\Norml}[1]{%
{|}\kern-.25ex{|}\kern-.25ex{|}#1{|}\kern-.25ex{|}\kern-.25ex{|}}

\title[]{Martingale Property of Exponential Semimartingales: A Note on Explicit Conditions and Applications to Financial Models}

\author[D. Criens]{David Criens}
\address{D. Criens - Technical University of Munich, Department of Mathematics, Germany}
\email{david.criens@tum.de}
\author[K. Glau]{Kathrin Glau}
\address{K. Glau - Technical University of Munich, Department of Mathematics, Germany}
\email{kathrin.glau@tum.de}
\author[Z. Grbac]{Zorana Grbac}
\address{Z. Grbac - Universit\'e Paris Diderot, Laboratoire de Probabilit\'es et Mod\`eles Al\'eatoires, France}
\email{grbac@math.univ-paris-diderot.fr}



\thanks{The authors thank Ernst Eberlein for useful discussions. The financial support from the DFG project EB66/9-3 and from the Entrepreneurial Program, Women for Math Science, TU Munich is gratefully acknowledged.}

\date{\today}
\maketitle

\frenchspacing
\pagestyle{myheadings}

\begin{abstract}
We give a collection of explicit sufficient conditions for the true martingale property of a wide class of exponentials of semimartingales. We express the conditions in terms of semimartingale characteristics. This turns out to be very convenient in financial modeling in general. Especially it allows us to carefully discuss the question of well-definedness of semimartingale Libor models, whose construction crucially relies on a sequence of measure changes.
\end{abstract}

\section{Introduction}
Local martingales are the core object of stochastic integration. 
Thus they provide a natural access to time evolutionary stochastic modeling, which is a cornerstone of mathematical finance.
The fundamental theorem of asset pricing states that the absence of arbitrage is essentially equivalent to the local martingale property of discounted asset prices under some equivalent probability measure.
One important benefit of the true martingale property of discounted asset price processes is their use for density processes of a change of measure. 
In financial terms this corresponds to a change of numeraire. 
Since the seminal work of \citet{GemanElKarouiRochet95} this concept became indispensable for both computational and modeling aspects.
Often a change of numeraire facilitates option pricing by reducing complexity of computations. Moreover, it is a building stone of the construction of Libor market models introduced by \citet{BraceGatarekMusiela97} and \citet{MiltersenSandmannSondermann97}. More fundamentally, a change of measure connects historical and risk-neutral probability measures. 
On the other hand if the discounted asset price process is a strict local martingale, i.e. a local martingale which is not a true martingale, %then the fundamental price deviates from the market price. T
this is sometimes interpreted as financial bubble.
%{\color{red} 
%Including a financial bubble makes a major difference from a modeling point of view, since it comes along with a considerable increase of complexity:
%Various subtle issues arise in models with financial bubbles, some of which contrast classical results such as that prices of American calls may exceed their European counterparts and the put-call parity may fail. Let us, 
However, the definition and existence of financial bubbles critically depends on the specific notion of the market price, arbitrage and admissible strategies, see for example \citet{CoxHobson2005} and \citet{JarrowProtterShimbo2010}. 
In a typical modeling situation it is enjoyable to work with true martingales. 

Usually price processes are non-negative and therefore are modeled as exponentials
of semimartingales, which form a wide and flexible class of positive processes. 
One can
characterize the local martingales in this class by a drift condition.
It is, however, more involved to identify conditions for their true martingale property. 
In order to formulate the problem more precisely denote by $X$ an $\erd$-valued semimartingale and by $\lambda$ an $\erd$-valued predictable process which is integrable with respect to $X$. Then $\lambda \cdot X:=\sum_{i\leq d} \int_{0}^{\cdot} \lambda^{i} \ud X^{i}$ denotes the real-valued stochastic integral process of $\lambda$ with respect to $X$. %
Moreover, let $V$ be a predictable process with finite variation. 
We pose the following question: \emph{Under which conditions on the characteristics of $X$ is a
real-valued semimartingale $Z$ of the form
\[
Z:=e^{\lambda \cdot X - V}
\]
a (uniformly integrable) martingale?}

If \(e^{\lambda \cdot X}\) is a special semimartingale, there exists a unique predictable process of finite variation \(V\) such that \(Z\) is a local martingale.
In this case, $V$ is called the \emph{exponential compensator} of $\lambda\cdot X$, see Section 2 for details.
Various criteria for the more delicate true martingale property of $Z$  have been proposed.
The seminal paper by \citet{Novikov72} treats the continuous semimartingale case. Sufficient conditions for general semimartingales are provided for example in \citet{LepingleMemin78}, \citet{KallsenShiryaev02}, \citet{Jacod79}, \citet{CheriditoFilipovicYor2005} and \citet{ProtterShimbo08}, see also a recent paper by \citet{LarssonRuf14} for further generalizations of Novikov-Kazamaki type conditions based on convergence results for local supermartingales.  Moreover, we refer to Section 1 and Section 3 of \citet{KallsenShiryaev02} for an exhaustive literature overview. In the special case when $X$ is a process with independent increments and absolutely continuous characteristics and $\lambda$ deterministic,
%Proposition 4.4 in
\citet{EberleinJacodRaible05} show that if $Z$ is a local martingale, it is also a true martingale. Deterministic conditions ensuring the martingale property of an exponential of an affine process are given in \citet{KallsenMuhleKarbe10}.
The conditions for more general semimartingales are not as explicit. 

Our contribution is to give %derive, based on \citet{KallsenShiryaev02}, 
explicit conditions for the martingale property of an exponential quasi-left continuous semimartingale in terms of its characteristics.
In Section 2 we introduce the notation and describe the general semimartingale setting following \citet{JacodShiryaev03}.
%on the semimartingale characteristics of $X$ for quasi-left continuous semimartingales. 
Section 3 contains the main results.
The advantage of the explicit conditions is their convenience for applications. 
We illustrate this by investigating the true martingale property of asset prices in semimartingale stochastic volatility models in Section \ref{SVAPM}.
Finally, in Section \ref{s:Libor} we prove the well-definedness of the backward construction of L\'evy Libor models. More precisely, we show that the candidate density processes for the measure changes are indeed true martingales which has not been rigorously proved earlier. Moreover, we present a natural extension to the semimartingale Libor model.

\section{Semimartingale notation and preliminaries}
\label{eq:notation}
In this section we introduce the notation and summarize the basic notions and facts from the semimartingale theory in order to keep the paper self-contained. 
Our main reference is \citet{JacodShiryaev03}, whose notation we use throughout the paper. Other standard references for stochastic calculus and semimartingales are e.g. \citet{Jacod79}, \citet{Metivier82} and \citet{Protter04}.

Let  $(\Omega,\lijepof, (\lijepof_t)_{t\geq 0}, \P)$ denote a stochastic basis, i.e. a filtered probability space with right-continuous filtration. For a class of processes \(\mathcal{C}\), we say that a process \(X\) is in the localized class \(\mathcal{C}_\textup{loc}\) if there exits a sequence of stopping times \((\tau_n)_{n \in \mathbb{N}}\) such that a.s. \(\tau_n \uparrow \infty\) as \(n \to \infty\) and \(X^{\tau_n} \in \mathcal{C}\). Denote by \(\mathcal{M}\) the class of c\`adl\`ag uniformly integrable martingales. The processes in the localized class $\lijepom_{\loc}$ are called \emph{local martingales}.
We denote by $\lijepov^{+}$ (resp. $\lijepov$) the set of all real-valued  c\`adl\`ag processes starting from zero that have \emph{non-decreasing paths} (resp. \emph{paths with finite variation} over each finite interval $[0,t]$). Let $\lijepoa^{+}$ denote the set of all processes $A \in \lijepov^{+}$ that are integrable, i.e. such that $\E[A_{\infty}]< \infty$, where $A_{\infty}(\omega):=\lim_{t \to \infty} A_{t}(\omega)\in\overline{\er}_+$ for every $\omega\in\Omega$. Moreover, let $\lijepoa$ denote the set of all $A \in \lijepov$ that have integrable variation, i.e.  $\Var(A) \in \lijepoa^{+}$, where for every $t \geq 0$ and every $\omega\in\Omega$, $\Var(A)_{t}(\omega)$ is defined as the total variation of the function $s \mapsto A_{s}(\omega)$ on $[0,t]$.
A process $X$ is called a \emph{semimartingale} if it has a decomposition of
the form
\begin{equation}\label{canonicaldec}
X=X_{0}+M+A,
\end{equation}
where $X_{0}$ is finite-valued and $\lijepof_{0}$-measurable, $M \in \lijepom_{\loc}$ with $M_{0}=0$ and $A \in \lijepov$. If $A$ in decomposition \eqref{canonicaldec} is predictable, $X$ is called a \emph{special semimartingale} and the decomposition is unique.
A semimartingale is called \textit{quasi-left continuous} if a.s. \(\Delta X_\tau = 0\) on the set \(\{\tau < \infty\}\) for all predictable times \(\tau\).

Let $X$ be an $\erd$-valued semimartingale, i.e. each component of $X$ satisfies \eqref{canonicaldec}. Denoting by $\varepsilon_a$ the Dirac measure at point $a$, the random measure of jumps $\mu^{X}$ of $X$ is an integer-valued random measure of the form
$$
\mu^{X}(\omega; \de t, \de x) := \sum_{s\geq0} \indik_{\{\Delta X_{s}(\omega) \neq 0\}} \varepsilon_{(s, \Delta X_{s}(\omega))} (\de t, \de x).
$$
There is a version of the predictable compensator of \(\mu^X\), denoted by \(\nu\), such that the \(\mathbb{R}^d\)-valued semimartingale \(X\) is quasi-left continuous if and only if \(\nu(\omega, \{t\} \times \mathbb{R}^d) = 0\) for all \(\omega \in \Omega\), c.f. \citet{JacodShiryaev03}, Corollary II.1.19.

In general, \(\nu\) satisfies
\begin{equation}
 \label{compensator-property}
(|x|^{2} \wedge 1) \ast \nu \in \lijepoa_{\loc}.
\end{equation}
The semimartingale  $X$ admits a \emph{canonical representation}
$$
X=X_{0}+B(h)+X^{c}+(x-h(x)) \ast \mu^{X} + h(x) \ast (\mu^{X}-\nu),
$$
where $h:\erd \to \erd$ is a \emph{truncation function}, i.e. a function that is bounded and behaves like $h(x)=x$ around 0, $B(h)$ is a predictable $\erd$-valued process with components in $\lijepov$, and $X^{c}$ is the continuous martingale part of $X$.

Denote by $C$ the predictable $\er^{d} \otimes \mathbb{R}^d$-valued covariation process defined as $C^{ij}:=\la X^{i,c}, X^{j,c} \ra$. %(in other words, $C$ is the predictable quadratic variation of the continuous martingale $X^c$).
Then the triplet $(B(h), C, \nu)$ is called the \emph{triplet of predictable characteristics} of $X$ (or simply the \emph{characteristics} of $X$). %{\color {red} (I like more to give the precise reference here, but I do not have the book yet)}
It can be shown (see Proposition II.2.9 in \citet{JacodShiryaev03}) that  there exists a predictable process $A\in \lijepoa_{\loc}^{+}$ such that
$$
 B(h) = b(h)\cdot A, \quad C = c\cdot A,\quad \nu = A\times F\,,
$$
where $b(h)$ is a $d$-dimensional predictable process, $c$ is a predictable process taking values in the set of symmetric non-negative definite $d\times d$-matrices and $F$ is a transition kernel from $(\Omega \times \er_{+}, \lijepop)$ into $(\erd, \lijepob(\erd))$. Here $\lijepop$ denotes the predictable $\sigma$-field on $\Omega \times \er_{+} $. We call $(b(h),c,F; A)$ the \emph{triplet of differential} (or \emph{local}) \emph{characteristics} of $X$.
If $X$ admits the choice $A_{t}=t$ above, we say that $X$ has \emph{absolutely continuous} characteristics (or shortly AC) and call \(X\) an \textit{It\^o semimartingale}.

An important subclass of semimartingales is the class of It\^o semimartingales with independent increments. These processes are known as time-inhomogeneous L\'evy processes or as Processes with Independent Increments and Absolutely Continuous characteristics (PIIAC), see e.g. Section 2 in \citet{EberleinJacodRaible05}. The differential characteristics $(b(h),c,F)$ of a PIIAC $X$, for every truncation function $h$, are deterministic and satisfy the following integrability assumption: For every $T>0$
\begin{align}
\label{eq:cond-2-PIIAC}
\int_0^T \bigg( |b(h)_s| + \| c_s\| + \int_{\erd} (|x|^2 \wedge 1) F_s(\dx)\bigg) \ds < \infty,
\end{align}
where $\|\cdot\|$ denotes any norm on the set of $d\times d$-matrices. For every $t>0$, the law of $X_t$ is characterized by a L\'evy-Khintchine type formula for its characteristic function, see again Section 2 in \citet{EberleinJacodRaible05}. This property makes the class of PIIAC particularly suitable for applications.
The following definition and results on exponentials of semimartingales are given in Definition 2.12, Lemma 2.13 and Lemma 2.15 in \citet{KallsenShiryaev02}.
\begin{definition}
A real-valued semimartingale $Y$ is called \emph{exponentially special} if $\exp(Y-Y_{0})$ is a special semimartingale.
\end{definition}
\begin{remark}\label{rem-expocomp}
Let $Y$ be a real-valued semimartingale and denote by  $\nu^{Y}$ the compensator of the random measure of jumps of $Y$ and $h$ a truncation function.
\begin{itemize}
\item[(a)] The following statements are equivalent:
\begin{itemize}
 \item[(i)] $Y$ is an exponentially special semimartingale.
 \item[(ii)]$(e^{y}-1-h(y)) \ast \nu^{Y} \in \lijepov$.
 \item[(iii)]   $ e^{y} \indik_{\{y > 1\}} \ast \nu^{Y} \in \lijepov$.
\end{itemize}
\item[(b)]
If $Y$ is exponentially special, then it admits an \emph{exponential compensator}, i.e. there exists a predictable process $V \in \lijepov$ such that $\exp(Y-Y_{0}-V) \in \lijepom_{\loc}$.
\end{itemize}
\end{remark}

Let $X$ be an $\erd$-valued semimartingale with differential characteristics $(b(h),c,F;A)$ and $\lambda\in L(X)$, where $L(X)$ denotes the set of predictable processes integrable with respect to $X$, c.f. \citet{JacodShiryaev03}, page 207. Moreover, assume that $\lambda\cdot X$ is exponentially special.
Following \citet{JacodShiryaev03}, Section III.7.7a we  define the \emph{Laplace cumulant process}
\begin{align}
\label{eq:exp-comp-1}
\widetilde{K}^{X}(\lambda) & := \widetilde{\kappa}^{X} (\lambda) \cdot A,
\end{align}
where
\begin{align}
\label{eq:exp-comp-2}
\widetilde{\kappa}^{X}_{s} (\lambda) & := \la \lambda_{s}, b_{s} \ra + \frac{1}{2} \la \lambda_{s}, c_{s} \lambda_{s}\ra + \int \big(e^{\la \lambda_{s},x\ra}-1- \la \lambda_{s}, h(x)\ra\big) F_{s}(\de x),
\end{align}
and the \emph{modified Laplace cumulant process}  %$K^{X}(\lambda)$ is defined as 
$K^{X}(\lambda)  := \ln (\lijepoe(\widetilde{K}^{X}(\lambda)))$, where $\lijepoe$ denotes the stochastic exponential, and %we have
\begin{eqnarray}\label{laplace cumulant process}
K^{X}(\lambda) %& := &\ln (\lijepoe(\widetilde{K}^{X}(\lambda)))\\
& = & \widetilde{K}^{X}(\lambda)  + \sum_{s \leq \cdot} (\ln (1 + \Delta \widetilde{K}^{X}_{s}(\lambda)) - \Delta \widetilde{K}^{X}_{s}(\lambda)).
\end{eqnarray}

The following results are proved in Proposition III.7.14 and Theorem III.7.4 in  \citet{JacodShiryaev03}:
\begin{proposition}
\label{Laplace-cum}
Let $X$ be an $\erd$-valued semimartingale and $\lambda\in L(X)$ such that $\lambda\cdot X$ is exponentially special.
\begin{itemize}
\item[(i)] The modified Laplace cumulant process $K^{X}(\lambda)$ is the exponential compensator of $\lambda \cdot X$, i.e. the process $Z$ defined by
$$
Z:=\exp (\lambda \cdot X - K^{X}(\lambda))
$$
is a local martingale. %(Theorem 2.19 in \citet{KallsenShiryaev02}).
\item[(ii)] If $X$ is quasi-left continuous,
the Laplace cumulant process $\widetilde{K}^{X}(\lambda)$ and the modified Laplace cumulant process $K^{X}(\lambda)$ coincide, i.e. $K^{X}(\lambda)=\widetilde{K}^{X}(\lambda)$.% for $\lambda\in L(X)$ such that $\lambda\cdot X$ is a special semimartingale.
\end{itemize}
\end{proposition}
%\begin{remark}
%A consequence of Corollary \ref{corollary PII} is that for each \(\mathbb{R}^d\)-valued continuous PII local martingale \(W\) (e.g. Brownian motion) and each deterministic \(\lambda \in L(W)\) the stochastic exponential \(\mathcal{E}(\lambda \cdot W)\) is a martingale.
%\end{remark}
%
%The question rises whether such a statement is also possible for more general semimartingales.
%This is the subject of the following section.
%Corollary \ref{local_martPIIAC} answers this question for PIIAC semimartingales \(X\) and bounded \(\lambda \in L(X)\) in Corollary \ref{local_martPIIAC}. 
In the following section we give sufficient conditions for the martingale property of exponential semimartingales. 
\section{The Martingale Property of Exponential Semimartingales}
\label{s:main}
Integrability conditions ensuring the (UI) martingale property of a non-negative or positive local martingale were studied from many perspectives and in various levels of generality.
It started with the classical conditions of \citet{Novikov72} which applies to continuous exponential local martingales. 
A natural generalization included jumps was given in the seminal paper of \citet{LepingleMemin78}. 
Various related conditions are given by \citet{KallsenShiryaev02}. A profound overview on Novikov-type conditions as well as boundedness conditions is given in the monograph of \citet{Jacod79}.
In this section we collect conditions for exponential semimartingales and express them in terms of semimartingale characteristics.
Thanks to these expression we usually call these type of conditions \emph{predictable conditions}. Let us start with a Novikov-type integrability condition which is based on the main result of \citet{LepingleMemin78}.
We follow its statement given by \citet{Jacod79} as Corollary 8.44. 
\begin{proposition}
\label{le:martY}
Let \(Y\) be a real-valued quasi-left continuous semimartingale with characteristics \((B(h), C, \nu)\).
If  
\begin{enumerate}[label={\rm (A\arabic{*})}]
% \item\label{firstY}$|y| e^{y} \indik_{\{|y| >1\}} \ast \nu^{Y} \in \lijepov$, and
 \item\label{secondI0T}$ \E \left(\exp \left\{\frac{1}{2} C_{T} +
((y-1)e^{y}+1)\ast \nu_{T} \right\} \right) < \infty$ for every $T\ge0$,\\[-1.5ex]
\end{enumerate}
the process
$
M:=e^{ Y-K^{Y}(1)}
$
is a true martingale. Moreover, replacing condition \ref{secondI0T} with\\[-3ex]
\begin{enumerate}[label={\rm (A\arabic{*}\ensuremath{{}})}]
\stepcounter{enumi}
 \item\label{secondI01} $ %\sup\limits_{t \in \er_{+}} 
\E \left(\exp \left\{\frac{1}{2} C_{\infty} +
((y-1)e^{y}+1)\ast \nu_{\infty} \right\} \right) < \infty$,
\end{enumerate}
$M$ is a UI martingale.
\end{proposition}
\begin{proof}
Note that the characteristics of \(X^T\), for any \(T > 0\), are given by \((B^T, C^T, \nu^T)\), where \(\nu^T(\de t, \de x) := \indik_{[0, T] \times \mathbb{R}^d} \nu (\de t, \de x)\). 
Now, since local martingales whose localizing sequence is deterministic are martingales, the first claim follows immediately from the second. 
Note that \ref{secondI0T} implies that \(X\) is exponentially special and hence that \(M\) is a local martingale. 
In view of Theorem 2.19 in \citet{KallsenShiryaev02}, the second claim follows from \citet{Jacod79}, Corollary 8.44.
\end{proof}
\begin{remark}
If \(Y\) is continuous, i.e. $\nu=0$, then \ref{secondI01} reduces to the classical Novikov condition as presented in Section 3.5.D in \citet{KaratzasShreve}.
\end{remark}
As an immediate corollary we derive the follows sufficient conditions for the case where \(X\) is given as a stochastic integral.
\begin{corollary}
\label{local_martingale}
Let $X$ be an $\erd$-valued quasi-left continuous semimartingale with differential characteristics $(b(h),c,F;A)$ and $\lambda \in L(X)$. % bounded. %and such that $\lambda\not\equiv0$.
If
\\[-3ex]
\begin{enumerate}[label={\rm (B\arabic{*})}]
% \item\label{eq:general-case1}
%$|x| e^{\la \lambda, x \ra} \indik_{\{|x|>1\}} \ast \nu \in \lijepov$, and %\\[-2ex]
 \item\label{eq:general-case2-b} for every $T\ge0$,
\begin{center}\(
\hspace{-0.75cm} \E\big(\exp \big\{\frac{1}{2} \int_{0}^{T} \la \lambda_{s}, c_{s} \lambda_{s} \ra \de A_{s} + \int_{0}^{T} \int_{\erd} ((\la \lambda_{s}, x\ra -1) e^{\la \lambda_{s}, x\ra}+1)F_{s}(\de x) \de A_{s} \big\}\big)< \infty
\)\end{center}
%$\\[0.2ex]
%for every $T\ge0$,\\[-2ex]
\end{enumerate}
the process
$
M:=e^{\lambda\cdot X-K^{X}(\lambda)}
$
is a true martingale. Moreover, replacing \ref{eq:general-case2-b} with\\[-2ex]
\begin{enumerate}[label={\rm (B\arabic{*}\ensuremath{{}})}]
\stepcounter{enumi}
 \item\label{eq:general-case2} %assume
%\begin{align*}
%\sup\limits_{t \in \er_{+}} 
\(\E\big(\exp\big\{\!\frac{1}{2} \!\int_{0}^{\infty} \la \lambda_{s}, c_{s} \lambda_{s} \ra \de A_{s} \! + \int_{0}^{\infty} \int_{\erd} ((\la \lambda_{s}, x\ra -1) e^{\la \lambda_{s}, x\ra}\!+1)F_{s}(\de x) \de A_{s}\big\}\big) < \infty,
\)
%\end{align*}
%$,\\[-1.5ex]
\end{enumerate}
then $M$ is a UI martingale.
\end{corollary}
\begin{proof}
The characteristics of \(\lambda \cdot X\) are given by Proposition IX.5.3 in \citet{JacodShiryaev03}. Now the claim follows by an application of Proposition \ref{le:martY}.
\end{proof}
\begin{remark}
Obviously, by considering an real-valued semimartingale \(X\) and \(\lambda = 1\) in Corollary \ref{local_martingale} we recover Proposition \ref{le:martY}.
\end{remark}
For applications the following boundedness condition turns out to be useful, see for instance Corollary \ref{coro 4.4}, Proposition \ref{prop-liboruimart} and \ref{r:semimartingale-Libor}  in the sections below.
\begin{proposition}
\label{local_martWithConst}
Let $X$ %and $\lambda$ 
be as in Corollary \ref{local_martingale} and let \(\lambda \in L(X)\). % be bounded. 
If
\\[-2ex]
\begin{enumerate}[label={\rm (C\arabic{*})}]
 \item\label{eq:condition-Libor-case1T} for every $T\ge0$, there exists a non-negative constant $\kappa(T)$ such that a.s.\\[-2.5ex]
\begin{eqnarray*}
\int_{0}^{T}  \la \lambda_{s}, c_{s} \lambda_{s}\ra  \de A_{s} + \int_0^T \int_{\mathbb{R}^d} \left(1 - \sqrt{e^{\langle \lambda_s, x\rangle}}\ \right)^2 F_s(\de x)\de A_s \leq \kappa (T) 
% \int_{0}^{T} \int_{\erd}\left[(|x|^{2}\wedge 1) + |x| e^{\la \lambda_s, x\ra}\indik_{\{|x|>1\}} \right] F_s(\dx)\de A_{s}  < \kappa(T), %\ \  \textrm{(a.s.)},
\end{eqnarray*} \\[-4.5ex]
\end{enumerate}
the process
$
M:=e^{\lambda\cdot X-K^{X}(\lambda)}
$
is a true martingale. Moreover, replacing \ref{eq:condition-Libor-case1T}
with\\[-2ex]
\begin{enumerate}[label={\rm (C\arabic{*}\ensuremath{{}})}]
\stepcounter{enumi}
 \item\label{eq:condition-Libor-case1} there exists a non-negative constant $\kappa$ such that a.s.
\begin{eqnarray*}
\int_{0}^{\infty}  \la \lambda_{s}, c_{s} \lambda_{s}\ra \de A_{s}  + \int_{0}^{\infty} \int_{\erd} \left(1 - \sqrt{e^{\langle \lambda_s, x\rangle}}\ \right)^2 F_s(\de x)\de A_s \leq \kappa
%\left[(|x|^{2}\wedge 1) + |x| e^{\la \lambda_s, x\ra}\indik_{\{|x|>1\}} \right] F_s(\dx)\de A_{s} < \kappa, %\ \ \textrm{(a.s.)},
\end{eqnarray*}\\[-4.5ex]
\end{enumerate}
then $M$ is a UI martingale.
\end{proposition}
\begin{proof}
Again, the first part is an immediate consequence of the second.
Note that \ref{eq:condition-Libor-case1T} implies that \(\lambda \cdot X\) is exponentially special.
Hence, we can deduce the claim from Theorem 2.19 in \citet{KallsenShiryaev02} together with Lemma 8.8 and Theorem 8.25 in \citet{Jacod79}.
\end{proof}
\begin{remark}
Clearly, thanks to Corollary \ref{local_martingale}, the condition
\ref{eq:condition-Libor-case1T} could be replaced by the following condition: for every \(T \geq 0\) there exists a constant
$\kappa(T)$ such that a.s.\begin{align}\label{compare}
\int_{0}^{T} \la \lambda_{s}, c_{s} \lambda_{s} \ra \de A_{s} + \int_{0}^{T} \int_{\erd} \left((\la \lambda_{s}, x\ra -1) e^{\la \lambda_{s}, x\ra}+1\right)F_{s}(\de x) \de A_{s} \leq \kappa(T).
\end{align}
The elementary inequality
\begin{align*}
0 \leq (1 - \sqrt{x})^2 \leq x \log (x) - (x - 1),\textup{ for all } x > 0,
\end{align*}
as for instance noted in \citet{Esche}, Lemma 2.13, shows that condition \ref{eq:condition-Libor-case1T} is an improvement to \eqref{compare}.
\end{remark}
Let us shortly turn to the subclass of semimartingales with independent increments (SII processes), %which can be described as the semimartingales with deterministic characteristics.
for which the situation is slightly different than in the more general case.
For exponential SII processes 
the  local martingale property is equivalent to the true martingale property.
From a mathematical finance perspective this interesting fact for instance implies that exponential SII models cannot include bubbles which are modeled as strict local martingales. 

Let us formalize this observation and add some simple deterministic conditions for the martingale property. 
%result shows that the martingale property of \(e^{\lambda \cdot X - K^Y(\lambda)}\), where \(X\) is a semimartingale with deterministic characteristics and determistic \(\lambda \in L(X)\), is implied by their local martingale property.
The main implication \((ii) \Rightarrow (i)\) is essentially thanks to \citet{KallsenMuhleKarbe10}, Proposition 3.12. %, which proof is based on a change of measure. 
Note that the assertion does not require quasi-left continuity.

\begin{proposition}\label{PII local martingale}
Let \(X\) be an \(\mathbb{R}^d\)-valued semimartingale with deterministic characteristics \((B^X, C^X, \nu^X)\), \(\lambda \in L(X)\) be deterministic and \(M := e^{\lambda \cdot X - K^X(\lambda)}\). %\bl{If \(e^{\langle \lambda, x\rangle} \indik_{\{\langle \lambda, x\rangle > 1\}} * \nu \in \mathcal{V}\), then \(M := e^{\lambda \cdot X - K^X(\lambda)}\) is a martingale.}
The following are equivalent
\begin{enumerate}
\item[(i)] \(M\)  is a martingale.
\item[(ii)] \(M\) is a local martingale.
\item[(iii)] \(\lambda \cdot X\) is exponentially special.
\item[(iv)] \((e^{\langle \lambda, x\rangle} - 1 - h(\langle \lambda, x\rangle)) * \nu^X \in \mathcal{V}\).
\item[(v)] \(e^{\langle \lambda ,x\rangle} \indik_{\{\langle \lambda, x\rangle > 1\}} * \nu^X \in \mathcal{V}\).
\end{enumerate}
\end{proposition}
\begin{proof}
The implication \((i)\) $\Rightarrow$ \((ii)\) is trivial and equivalence \((ii)\) $\Leftrightarrow$ \((iii)\) holds by definition. 
The equivalences of  \((iii),(iv)\) and \((v)\) are due to Remark \ref{rem-expocomp} and \citet{JacodShiryaev03}, Proposition IX.5.3, which shows that \(\lambda \cdot X\) has deterministic characteristics with
\begin{align*}
\nu^{\lambda \cdot X}(A, \de s) = \int_{\mathbb{R}^d} \indik_A\big(\langle \lambda_s, x\rangle\big) \nu^X(\de x, \de s),\quad A \in \mathcal{B}(\mathbb{R}\backslash \{0\}).
\end{align*}
It is left to show the implication \((iii)\) $\Rightarrow$ \((i)\).
%Note that, due to \citet{JS} Theorem III.7.4, that 
%\begin{align*}
%\Delta \big[ \lambda \cdot X - K^Y(\lambda)\big]_t = \lambda_t \Delta X_t - \log(1 + \Delta \widetilde{K}^X (\lambda)_t),
%\end{align*}
%where \(\Delta K^X(\lambda)_t = \int_{\mathbb{R}^d} (e^{\langle \lambda_t, x\rangle} - 1)\nu(\{t\}, \de x).\)
%Therefore we obtain from \citet{JS} Theorem II.8.10 that 
%\begin{align*}
%\end{align*}
%\y{We first deduce from \citet{JacodShiryaev03}, Proposition IX.5.3 that \(\lambda \cdot X\) has deterministic characteristics.} \bl{Moreover, in view} \g{In view} 
In view of \eqref{laplace cumulant process} and since $\lambda\cdot X$ has deterministic characteristics, \(K^X(\lambda)\) is a deterministic process of finite variation and hence also has deterministic characteristics. 
Define \(f(x, y) := x - y\), then \(Y := \lambda \cdot X - K^X(\lambda) = f(\lambda \cdot X, K^X(\lambda))\).
It follows from \citet{GollKallsen00}, Corollary 5.6 applied to \(f(\lambda \cdot X, K^X(\lambda))\) that \(Y\) has also deterministic characteristics.
From the relationship of ordinary and stochastic exponentials given in \citet{JacodShiryaev03}, Theorem II.8.10, 
 we obtain that 
\(
M = e^{Y} = \mathcal{E}(\overline{Y}),
\)
where \(\overline{Y}\) is a semimartingale with \(\Delta \overline{Y} = (e^{\Delta Y} - 1) > -1\). 
We deduce from \citet{JacodShiryaev03}, Equation II.8.14 that \(\overline{Y}\) inherits the property of deterministic characteristics from \(Y\).
Due to Remark \ref{rem-expocomp}(b), the condition \(e^{\langle \lambda, x\rangle} \indik_{\{\langle \lambda, x\rangle > 1\}} * \nu \in \mathcal{V}\) yields that \(\lambda \cdot X\) is exponentially special. Thus, since \(K^X(\lambda)\) is the exponential compensator of \(\lambda \cdot X\), c.f. Proposition \ref{Laplace-cum} (i), \(M\) is a local martingale.
The claim now follows from Proposition 3.12 in \citet{KallsenMuhleKarbe10}.
%, which states that \(M\) is a martingale if and only if it is a local martingale.}
\end{proof}

%\section{Application to financial modeling}

%In this section we present two applications of the results from Section \ref{s:main} to financial modeling. The monographs by \citet{Shiryaev99}, \citet{MusielaRutkowski05}, \citet{ContTankov03} and \citet{JeanblancYorChesney09} and the references therein provide a detailed overview concerning applications of semimartingales in finance.

%\subsection{Asset price models with stochastic volatility and stochastic interest rate}
\section{Applications to Finanical Models}
In this section we present two applications of the results from Section \ref{s:main} to financial modeling. A detailed overview concerning applications of general semimartingales in finance is for instance provided by the monographs of \citet{Shiryaev99}, \citet{ContTankov03}, \citet{MusielaRutkowski05} and \citet{JeanblancYorChesney09}. % and the references therein.
\subsection{Stochastic Volatility Asset Price Model}\label{SVAPM}
Here, we illustrate how the conditions of Section \ref{s:main} can be used to facilitate pricing in arbitrage-free models driven by semimartingales. % to a pricing framework.
Let $(\Omega, \lijepof, (\lijepof_t)_{0 \leq t \leq T}, \P)$ be a stochastic basis, where $T >0$ denotes a finite time horizon. We model the asset price $S$ and a bank account $B$ with stochastic interest rate $r$ by
\begin{equation}\label{semimartmodel}
S:= S_0e^{\sigma^S\cdot X^S - V},\qquad  B:=e^{\sigma^r\cdot X^r}
\end{equation}
with $S_0>0$, a $d$-dimensional semimartingale $X:=(X^S,X^r)$ with $X^S$ $d_1$-dimensional and $X^r$ $d_2$-dimensional such that  $d_1+d_2=d$,
and a $d$-dimensional predictable process $\sigma:=(\sigma^S,-\sigma^r)$ with $\sigma^S \in L(X^S)$ and $\sigma^r \in L(X^r)$  such that $\sigma\cdot X$ is exponentially special.
We assume that the process $V$ is the exponential compensator of \(\sigma^S \cdot X^S - \sigma^r \cdot X^r\), c.f. Proposition \ref{Laplace-cum}.
Thanks to this assumption, the discounted stock price \(\widetilde{S} := B^{-1} S\) is a local martingales, i.e. in other words \(\P\) is a risk-neutral probability measure.
According to the fundamental theorem of asset pricing for general semimartingales in \citet{DelbaenSchachermayer98}, the No Free Lunch With Vanishing Risk (NFLVR) holds is this case.
%condition is implied by the local martingale property of \(\widetilde{S}\) under \(\P\).
%In other words, model \eqref{semimartmodel} is a semimartingale pricing model specified directly under a risk-neutral probability measure $\P$. Note that in general the risk-neutral probability measure may not be unique and the model is incomplete. 

Note that in general the risk-neutral probability measure may not be unique and the model is incomplete. Let us now consider a European call option with strike $K>0$ with payoff
$(S_T-K)^+$ at maturity $T>0$. Its fundamental price under $\P$, denoted by $C^\ast_t$  for any $t\in [0, T]$,  is given by
\begin{equation}\label{callprice}
C^\ast_t := B_t\E_{\P}\big(B_T^{-1} (S_T-K)^+\big|\lijepof_t\big) \le B_t\E_{\P}\big(\widetilde{S}_T \big|\lijepof_t\big) \le  B_t \widetilde{S}_t  = S_t <\infty,
\end{equation}
which is well-defined and finite a.s. The inequality $\E(\widetilde{S}_T|\lijepof_t)\le \widetilde{S}_t$ is a consequence of $\widetilde{S}$ being a positive local martingale and hence a supermartingale.
%\textcolor{red}{What if \(S_0 < 0\), then the inequality in converse}
%which is well-defined because of $B_T^{-1}\max(S_T-K,0) \le \widetilde{S}_T$ and $\E(\widetilde{S}_T|\lijepof_t)\le \widetilde{S}_t<\infty$. The latter inequality is a consequence of $\widetilde{S}$ being a positive local martingale and therefore a supermartingale. 
The price $C^\ast$ is an arbitrage-free price, which -- even in the case of a complete
model -- might be non-unique. This subtle issue is closely related to financial bubbles, see for example Definition 3.6 in \citet{JarrowProtterShimbo2010} and Definition 2.10 in \citet{BiaginiFoellmerNedelcu14}. For a detailed mathematical treatment we refer to \citet{Protter13}.

When $\widetilde{S}$ is a true martingale, these delicate issues do not appear:  The asset price has no bubble and the market prices coincide with the fundamental prices -- this was proved for example in the setting of \citet{JarrowProtterShimbo2010}. 
Thus, our Section \ref{s:main} provides convenient conditions  to exclude ambiguities in the pricing due to the possible presence of bubbles.

To illustrate a further benefit of the explicit martingale conditions from Section \ref{s:main}  we use the true martingale property of the discounted asset price to perform a change of numeraire which reduces the complexity of a pricing problem at hand.
Considering for example a call option as above, in order 
to compute the expectation in \eqref{callprice} directly, information on the joint distribution of $S$ and $B$ is required. 
Here the true  martingale property of $\widetilde{S}$ allows to facilitate the computation of the expectation by a change of numeraire. More precisely, we can express the call price as a conditional expectation of a function of the asset value $S_T$ solely. Defining a probability measure $\widetilde{\P}$ via $\frac{\ud \widetilde{\P}}{\ud \P}|_{\lijepof_t}:=S_0^{-1}\widetilde{S}_t$ for $0\le t\le T$, and denoting by $\E_{\widetilde{\P}}$ the expectation under $\widetilde{\P}$, Bayes formula yields%, {\color{red} remove? would like to keep it: see e.g. Lemma A.1.4 in \citet{MusielaRutkowski05}}, yields
\begin{equation}\label{callpricenew}
C_t=S_t \E_{\widetilde{\P}}\big((1-KS_T^{-1})^+ \big|\lijepof_t\big).
\end{equation}
Compared with the original pricing formula 
$C_t = B_t\E_{\P}\big(B_T^{-1} (S_T-K)^+\big|\lijepof_t\big)$, %
the random variable \(B_T\) does not appear in the conditional expectation in \eqref{callpricenew}. This typically facilitates the computation since the semimartingale characteristics of \(S\) are known under the new probability measure.

By combining Corollary \ref{local_martingale} and Proposition \ref{local_martWithConst} the following characterization of the true martingale property for the semimartingale asset price model defined above.
\begin{corollary}\label{coro 4.4}
%Let a semimartingale asset price model \eqref{semimartmodel} be given with  $\sigma$ bounded and $V$ the exponential compensator of $\sigma \cdot X$. 
Assume that \(X\) is quasi-left continuous and denote its local characteristics by $(b,c,F;A)$. % the local semimartingale characteristics of $X$.
If $(b,c,F;A)$ and $\sigma$ satisfy %\ref{eq:general-case1} and 
 \ref{eq:general-case2-b} of Corollary~\ref{local_martingale}, resp. condition \ref{eq:condition-Libor-case1T} of Proposition \ref{local_martWithConst}, %or\\[-2ex]
the discounted asset price process $\widetilde{S}$ is a true martingale. 
\end{corollary}
Under the conditions of Corollary \ref{coro 4.4} the fair price at time \(t\) of the call option with maturity $T$ and strike $K$ is therefore given by \(C_t\) in \eqref{callpricenew}. 

\subsection{Semimartingale Libor model}
%\subsection{Libor models}
\label{s:Libor}
In this subsection we apply the results from Section \ref{s:main} to Libor models. These are models for discretely compounded forward interest rates known as Libor rates, where the term Libor stems from the London Interbank Offered Rate. The Libor models were introduced in \citet{BraceGatarekMusiela97} and \citet{MiltersenSandmannSondermann97} and later further developed and studied by many authors. We refer to \citet{MusielaRutkowski05}, Section 12.4, for a detailed overview. 

The challenge in modeling Libor rates is to simultaneously define the rates for different maturities as local martingales under different equivalent measures which ensures the absence of arbitrage.
These measures are in fact forward measures and they are interconnected via the Libor rates themselves. 
A convenient way to obtain such a model is by backward construction, following the pioneering work of \citet{MusielaRutkowski97b}.
This construction relies on the martingale property of Libor rates (under the corresponding forward measures), which allows to define changes of measure.
In the backward construction the Libor rates thus have to be not only local, but true martingales under their corresponding forward measures. When the model is driven by a continuous semimartingale this is standard by using Novikov type conditions, but verifying that the Libor rates are true martingales in general semimartingale models including jumps is more involved and has not been properly addressed in the financial literature.
Using explicit conditions from Section \ref{s:main}, we study this issue in detail below to close this gap.

Let us begin by describing a general semimartingale Libor model. Assume that $T^{*}>0$ is a fixed finite time horizon and we are given a pre-determined collection of maturities
 $0=T_{0} < T_{1} < \ldots < T_{n}=T^{*}$, with $\delta_k := T_{k+1} - T_k$ for \(k = 0, \ldots, n-1\). Moreover, let $(\Omega, \lijepof_{T^{*}}, (\lijepof_{t})_{0 \leq t \leq T^{*}}, \P_{T^{*}})$ be a stochastic basis. 
A general semimartingale Libor model consists of a family of semimartingales modeling the Libor rates \((L(\cdot, T_k))_{1 \leq  k \leq n-1}\) for lending periods \(([T_k, T_{k+1}])_{1 \leq k \leq n-1}\) and a family of probability measures \((\P_{T_k})_{1 \leq k \leq n}\), where \(L(\cdot, T_k)\) and \(\P_{T_k}\) are defined on \((\Omega, \mathcal{F}_{T_k}, (\mathcal{F}_t)_{0 \leq t \leq T_k})\) and \(\P_{T_n} = \P_{T^*} \), such that
\begin{enumerate}[leftmargin=3.9em, label=(A\arabic{*}*),widest=(A4)]
\item[(SML1)] \(L(\cdot, T_k)\) is a \(\P_{T_{k+1}}\)-martingale for all \(k = 1, \ldots, n-1\).
\\[-1.3ex]
\item[(SML2)] For \(k= n-1,\ldots, 1\)
%The family of probability measure is related by the Radon-Nikodym derivative
\begin{align*}
\frac{\de \P_{T_k}}{\de \P_{T_{k+1}}} = \frac{1 + \delta_k L(T_k, T_k)}{1 + \delta_k L(0, T_k)}. %\quad \delta_k := T_{k+1} - T_k.
\end{align*}
\end{enumerate}

For each \(k\) the probability measure \(\P_{T_k}\) is called the \textit{forward Libor measure} for maturity \(T_k\), cf. \citet{MusielaRutkowski05}, Definition 12.4.1. The measure \(\P_{T_k}\) is  in fact the forward martingale measure associated with maturity  \(T_k\) and the density process above is a forward price process. This can be seen from the link between forward Libor rates and zero-coupon bond prices, see \citet{MusielaRutkowski05}, Sections 12.1.1 and 12.4.4.

Below we present the main ideas of the backward construction of the Libor model in a semimartingale framework. We start by modeling the Libor rate with the most distant maturity under a given probability measure and then proceed backwards. We define in each step the next forward measure via a density process based on the previously modeled Libor rates and model the next Libor rate under this measure.

Let  $X$ be an $\erd$-valued  semimartingale on the stochastic basis  $(\Omega, \lijepof_{T^{*}}, (\lijepof_{t})_{0 \leq t \leq T^{*}}, \P_{T^{*}})$. We  start by modeling the Libor rate \(L(\cdot, T_{n-1})\) for maturity \(T_{n-1}\) by
\begin{align}\label{construction step 1}
L(t, T_{n-1}) := L(0, T_{n-1}) \exp\big\{ \lambda(\cdot, T_{n-1}) \cdot X_t - K^X_t(\P_{T^*}, \lambda(\cdot, T_{n-1}))\big\}, 
\end{align}
for $t\leq T_{n-1}$, 
where \(L(0, T_{n-1}) > 0\) and \(\lambda(\cdot, T_{n-1}) \in L(X)\) is a volatility process such that the stochastic integral \(\lambda(\cdot, T_{n-1}) \cdot X\) is \(\P_{T^*}\)-exponentially special with \(K^X(\P_{T^*},\lambda(\cdot, T_{n-1}))\) its \(\P_{T^*}\)-exponential compensator. Hence,  \(L(\cdot, T_{n-1})\) is a \(\P_{T^*}\)-local martingale.  
Assuming that \(L(\cdot, T_{n-1})\) is a true \(\P_{T^*}\)-martingale, we can define the probability measure \(\P_{T_{n-1}}\) on \((\Omega, \mathcal{F}_{T_{n-1}})\) by the Radon-Nikodym derivative 
\begin{align}
\frac{\de \P_{T_{n-1}}}{\de \P_{T^*}} = \frac{1 + \delta_{n-1} L(T_{n-1}, T_{n-1})}{1 + \delta_{n-1} L(0, T_{n-1})}.
\end{align}
%where \(\delta_{n-1} := T^* - T_{n-1}\).
Moreover, we obtain for \(t \leq T_{n-1}\) 
\begin{align}
\frac{\de \P_{T_{n-1}}}{\de \P_{T^*}}\bigg|_{\mathcal{F}_t} = \frac{1 + \delta_{n-1} L(t, T_{n-1})}{1 + \delta_{n-1} L(0, T_{n-1})}.
\end{align}
Now we recursively model the Libor rates  $L(\cdot, T_k)$ for \(k = n-2, \ldots, 1\) by
\begin{align}
\label{eq:Libor-rate}
L(t, T_k) := L(0, T_k)\exp \big\{\lambda(\cdot, T_{k}) \cdot X_t - K^{X}_t(\P_{T_{k+1}}, \lambda(\cdot, T_{k}))\big\},
\end{align}
for $t\leq T_k$, where \(L(0, T_k) > 0\) and \(\lambda(\cdot, T_k)\in L(X)\) is a volatility process such that \(\lambda(\cdot, T_k)\cdot X\) is \(\P_{T_{k+1}}\)-exponentially special with \(\P_{T_{k+1}}\)-exponential compensator \(K^X(\P_{T_{k+1}}, \lambda(\cdot, T_k))\). As above, this means that \(L(\cdot, T_{k})\) is a \(\P_{T_{k+1}}\)-local martingale.  Note that the Libor rate for the interval starting at $T_0=0$ and ending at $T_1$ is simply  the given spot Libor rate $L(0, T_0)>0$. 
The probability measure \(\P_{T_{k}}\) is defined on \((\Omega, \mathcal{F}_{T_k})\) by the Radon-Nikodym derivative
\begin{align}
\frac{\de \P_{T_{k}}}{\de \P_{T_{k+1}}} = \frac{1 + \delta_{k} L(T_{k}, T_{k})}{1 + \delta_{k} L(0, T_{k})}, %\quad \delta_k := T_{k+1} - T_k,
\end{align}
where it has to be assumed that \(L(\cdot, T_{k})\) is a true \(\P_{T_{k+1}}\)-martingale. Then we have for $t\leq T_k$
\begin{align}
\frac{\de \P_{T_{k}}}{\de \P_{T_{k+1}}}\bigg|_{\mathcal{F}_t} = \frac{1 + \delta_{k} L(t, T_{k})}{1 + \delta_{k} L(0, T_{k})}.
\end{align}
Furthermore, we obtain that the probability measure \(\P_{T_{k+1}}\) is related to \(\P_{T^*}\) via 
\begin{eqnarray}
\label{dPTk}
%\label{dPTk}
%\P_{T_n}=\P_{T^\ast},\quad\text{and }
\frac{\ud \P_{T_{k+1}}}{\ud \P_{T^{*}}}\bigg|_{\lijepof_{t}}&=&\prod_{i=k+1}^{n-1}\frac{ 1+ \delta_i L(t,T_i)}{ %\prod_{i=k}^{n-1} 
 1+ \delta_i L(0,T_i)},%\quad\text{ for }k< n-1\,\text{and }t\le T_{k+1}.
\quad t \leq T_{k+1}.
\end{eqnarray}
Note that the construction is well-defined if the Libor rates \(L(\cdot, T_k)\) are \(\P_{T_{k+1}}\)-martingales for all \(k = 1,\ldots, n-1\).

To justify the backward construction \eqref{construction step 1} -- \eqref{dPTk} of the measures $(\P_{T_{k}})_{1 \leq k \leq n-1}$, we prove the required martingale property of the Libor rates in the proposition below.
%Therefore the convenient condition of Proposition \ref{PII local martingale}, namely that local martingality is equivalent to martingality, does not apply.

\begin{proposition}
\label{r:semimartingale-Libor}
Let $X$ in equation \eqref{eq:Libor-rate} be an  $\erd$-valued quasi-left continuous semimartingale with differential characteristics $(b^{T^\ast}, c, F^{T^\ast};A)$ with respect to $P_{T^\ast}$, and non-negative $\lambda(\cdot, T_{k}) \in L(X)$.
%nonnegative and bounded such that $\lambda(t, T_{k}) = 0$, for $T_k <t \le T^*$. 
Assume \\[-3ex]%In this case, assumptions \ref{EM}-\ref{det} are replaced with the following assumption:\\[-3ex]
\begin{itemize}
\item[\rm{(SL)}] for all $i=1, \ldots, n-1$ there exists a non-negative constant $\kappa$ such that a.s.\\[-3ex]
\begin{align*}
\int_{0}^{T^*}  \la& \lambda(t, T_{i}), c_{t} \lambda(t, T_{i})\ra  \de A_{t} 
\\&+ \int_{0}^{T^*} \int_{\erd}\left(1- \sqrt{e^{\la \lambda(t, T_{i}), x\ra}}\ \right)^2 e^{\la \sum_{k=i+1}^{n-1} \lambda(t, T_{k}), x\ra}\ F^{T^\ast}_t( \dx)\de A_t
% & \qquad \qquad \qquad \qquad \qquad \qquad\qquad  +  \int_{0}^{T^*}  \la \lambda(t, T_{i}), c_{t} \lambda(t, T_{i})\ra  \de A_{t} 
\leq \kappa, %\ \ \textrm{(a.s.)},\\[-4ex]
\end{align*}
where we use the convention \(\sum_{\emptyset} = 0\).
\end{itemize}
%then the claim of Proposition \ref{prop-liboruimart} is valid.
%\\
%OR\\
Then for each $k=1,\ldots,n-1$, the process $L(\cdot,T_k)$ from \eqref{eq:Libor-rate} is a martingale with respect to $\P_{T_{k+1}}$ given by \eqref{dPTk}.
\end{proposition}
\begin{proof}
%The proof follows along the same lines as the proof of Proposition~\ref{prop-liboruimart}. 
For $k=n-1$, the assertion follows directly from assumption (SL) and Proposition~\ref{local_martWithConst}.

For $k \leq n-2$, denote the semimartingale characteristics of $X$ with respect to $\P_{T_{k+1}}$ by $(B^{T_{k+1}}, C^{T_{k+1}}, \nu^{T_{k+1}})$. 
Next we compute these characteristics by backward induction and Girsanov's theorem as given by Theorem III.3.24 in \citet{JacodShiryaev03}.
We shortly give some details on the application of Girsanov's theorem.
Denote \(\de \P_{T_{n-1}}/\de \P_{T^*}|_{\mathcal{F}_\cdot} =: Z^{T_{n-1}}\) and note that 
\[L(\cdot, T^*) = L(0, T_{n-1}) \mathcal{E}\left( \lambda(\cdot, T_{n-1}) \cdot X^{c, T^*} + \left(e^{\langle\lambda(\cdot, T_{n-1}), x\rangle} - 1\right) * \left(\mu^X - \nu^{T^*}\right)\right),\] 
where \(X^{c, T^*}\) denotes the continuous local \(\P_{T^*}\)-martingale part of \(X\).
We have 
\begin{align*}
M^{\P_{T^*}}_{\mu^X}( Z^{T_{n-1}} |\widetilde{\mathcal{P}}) 
%&= Z^{T^{n-1}}_- + M^{\P_{T^*}}_{\mu^X}(\Delta Z^{T_{n-1}}|\widetilde{\mathcal{P}})
&= Z^{T^{n-1}}_- + \frac{\delta_{n-1} L(\cdot -, T_{n-1})}{1 + \delta_{n-1}L(0, T_{n-1})} \left(e^{\langle \lambda(\cdot, T_{n-1}, x\rangle} - 1\right).
%\frac{1}{1 + \delta_{n-1}L(0, T_{n-1})}+ \frac{\delta_{n-1}L(\cdot -, T_{n-1})}{1 + \delta_{n-1}L(0, T_{n-1})} \left(e^{\langle \lambda (\cdot, T_{n-1}),x\rangle} - 1\right)
\end{align*}
Now Girsanov's theorem yields that 
\begin{align*}
C^{T_{n-1}} = C,\quad \nu^{T_{n-1}} (\de t, \de x) = \beta(t, x, T_{n-1}) F^{T^*}_t (\de x) \de A_t,
\end{align*}
with
$$
\beta(t,x,T_{l}) := \frac{\delta_{l}L(t-,T_{l})}{1+\delta_{l}L(t-,T_{l})} \left(e^{\la \lambda(t, T_{l}),x \ra}-1\right) +1,
$$
for $l=1, \ldots, n-1$.
Repeating these steps, by backward induction we obtain 
%, as in \citet{EberleinOezkan05}, page 341-342, we show
\begin{align*}
C^{T_{k+1}}  =C ,\quad
%\label{eq:nu-under-forward-measure}
\nu^{T_{k+1}}(\de t, \de x)  = \prod_{l=k+1}^{n-1} \beta(t,x,T_{l}) F^{T^{*}}_t( \de x)\de A_t.
\end{align*}
Noting that \(\beta(t, x, T_l) \leq e^{\langle \lambda(t, T_l), x\rangle}\), the claim follows from Proposition~\ref{local_martWithConst}.
\end{proof}
%For the setting of \citet{EberleinOezkan05}, additionally also allowing for stochastic volatility, we obtain the following:
Let us link our discussion to the L\'evy Libor model of \citet{EberleinOezkan05} in which the driving process $X$ %=(X_{t})_{t \in [0, T^{*}]}$
 is assumed to be an $\erd$-valued PIIAC %\y{on $(\Omega, \lijepof=\lijepof_{T^{*}}, \mathbb{F}=(\lijepof_{t})_{0 \leq t \leq T^{*}}, \P_{T^{*}})$}
 with differential characteristic  $(0, c, F^{T^{*}})$ under $\P_{T^{*}}$. 
Eberlein and \"Ozkan impose the following assumptions:
\textit{
For some $M, \varepsilon >0$ and every $k=1,\ldots,n-1$ we have
\begin{enumerate}[label={\rm (L\arabic{*})}]
 \item\label{EM}
$\int_{0}^{T^{*}} \int_{|x| >1} e^{\la u, x\ra} F^{T^{*}}_{t} (\de x) \de t < \infty$ for every $u \in [-(1+\varepsilon)M, (1+\varepsilon)M]^{d}$,
\\[-0.5ex]
\item\label{lambda} $\lambda(\cdot, T_{k}): [0, T^{*}] \to \er^{d}_{+}$ is a bounded, nonnegative function such that for $t >T_{k}$, $\lambda(t, T_{k})=0$ and $ \sum_{k=1}^{n-1} \lambda^{j}(t, T_{k}) \leq M,$  for all $t \in [0, T^{*}]$ and every coordinate $j \in \{1, \ldots, d\}$,\\[-0.5ex]%, where $M>0$ is a constant from $(\mathbb{EM})$.
\item\label{det} $\lambda(\cdot, T_{k}): [0, T^{*}] \to \er^{d}_{+}$ is deterministic.
\end{enumerate}
}
Let us point that even when the driving process has deterministic characteristics under \(\P_{T^{*}}\) and \(\lambda\) is deterministic (as in the case above), the characteristics of \(X\) under \(\P_{T_{k}}\) for \(k = 1, \ldots, n-1\) are stochastic. 

We obtain the following sufficient conditions for the L\'evy Libor model, where we also allow \(\lambda\) to be stochastic. 
%Keeping the boundedness assumption (L2), we obtain the following corollary of Proposition \ref{r:semimartingale-Libor}.
\begin{corollary}\label{prop-liboruimart}
Assume that \(\sum_{j = 1}^n |\lambda (\cdot, T^j)| \leq N\) for a non-negative constant \(N\), and that there exists a non-negative constant \(\kappa\) such that 
\[
\int_0^{T^*} \int_{|x| > 1} e^{N |x|} F^{T^*}_t (\de x) \de t \leq \kappa.
\]
Then for each $k=1,\ldots,n-1$ the process $L(\cdot,T_k)$ defined in \eqref{eq:Libor-rate} is a martingale with respect to $\P_{T_{k+1}}$ given by \eqref{dPTk}.%$\frac{d \P_{T_{k}}}{d \P_{T_{k+1}}}=\frac{L(\cdot,T_k)}{L(0,T_k)}$.
\end{corollary}
\begin{proof}
It suffices to show that (SL) is satisfied. 
Note that we find a non-negative constant \(K^*\) such that for any \(i = 1, ..., n-1\) for all \(x \in \mathbb{R}^d\) with \(|x| \leq 1\)
%on \(\{y \in \mathbb{R}^d : |y| \leq 1\}\)
\begin{align*}
 \left(1- \sqrt{e^{\la \lambda(t, T_{i}), x\ra}}\ \right)^2 &e^{\la \sum_{k=i+1}^{n-1} \lambda(t, T_{k}), x\ra} \leq K^* |x|^2.
\end{align*}
Next we bound the large jumps. 
Using the fact that \((1 - \sqrt{x})^2 \leq 1 + x\) for \(x > 0\) and some non-negative constant \(K\), and the Cauchy-Schwarz inequality, we obtain
\begin{align*}
\int_{0}^{T^*} \int_{|x| > 1}&\left(1- \sqrt{e^{\la \lambda(t, T_{i}), x\ra}}\ \right)^2 e^{\la \sum_{k=i+1}^{n-1} \lambda(t, T_{k}), x\ra}\ F^{T^\ast}_t( \dx)\de t
\\&\leq \int_0^{T^*} \int_{|x| > 1} \left(e^{\langle \sum_{k = i+1}^{n-1} \lambda(t, T_k), x\rangle} + e^{\langle \sum_{k = i}^{n-1} \lambda(t, T_k), x\rangle}\right) F^{T^*}_t (\de x) \de t
\\&\leq 2 \int_0^{T^*} \int_{|x| > 1} e^{N |x| } F^{T^*}_t (\de x) \de t.
\end{align*}
%where we use the Cauchy-Schwarz inequality and the fact that \((1 - \sqrt{x})^2 \leq K x\) for \(x > 0\). 
Finally, since 
\begin{align*}
\int_0^{T^*} \langle \lambda(t, T_i), c_t \lambda (t, T_i)\rangle \de t \leq N^2\int_0^{T^*}  \| c_t\| \de t,
\end{align*}
where \(\| \cdot \|\) denotes the operator norm of \(c\), we conclude that (SL) holds. 
This concludes the proof.
\end{proof}

As mentioned in the introduction of the section, the martingale property of Libor rates under their corresponding measures is crucial for the validity of the backward construction of Libor models.  Therefore, Proposition \ref{r:semimartingale-Libor} and Corollary \ref{prop-liboruimart} provide a theoretical justification of the construction of the L\'evy Libor model by \citet{EberleinOezkan05}, and more generally of Libor models driven by quasi-left continuous semimartingales.

%------------------------------
%------------------------------

\bibliographystyle{chicago}
%\bibliography{/home/famuser/papapan/Papers/references}
\bibliography{references1}

\end{document}